\documentclass[11pt]{article}

\usepackage[T1]{fontenc}

\usepackage{amsfonts}  
\usepackage{url}
\begin{document}
\title{Shallow Circuits with High-Powered Inputs}

\author{Pascal Koiran\\
LIP\thanks{UMR 5668 ENS Lyon, CNRS, UCBL, INRIA.}, \'Ecole Normale Sup\'erieure de Lyon, Universit\'e de Lyon\\
Department of Computer Science, University of Toronto\thanks{A part 
of this work was done during a visit to the Fields Institute.}\\
{\tt pascal.koiran@gmail.com} 
}

\maketitle

\newtheorem{conjecture}{Conjecture}
\newtheorem{theorem}{Theorem}
\newtheorem{lemma}{Lemma}
\newtheorem{proposition}{Proposition}
\newtheorem{corollary}{Corollary}
\newtheorem{definition}{Definition}
\newtheorem{problem}{Problem}
\newtheorem{remark}{Remark}
\newtheorem{example}{Example}
\newtheorem{hypothesis}{Hypothesis}

\makeatletter
\def\@yproof[#1]{\@proof{ #1}}
\def\@proof#1{\begin{trivlist}\item[]{\em Proof#1.}}
\newenvironment{proof}{\@ifnextchar[{\@yproof}{\@proof{} 
}}{~$\Box$\end{trivlist}}
\makeatother

\newcommand\cc{\ensuremath{\mathbb{C}}}

\newcommand\vpspace{\ensuremath{\mathsf{VPSPACE}}}
\newcommand\pspace{\ensuremath{\mathsf{PSPACE}}}
\newcommand\vpzero{\ensuremath{\mathsf{VP}^0}}
\newcommand\rr{\ensuremath{\mathbb{R}}}
\newcommand\vp{\ensuremath{\mathsf{VP}}}
\newcommand\vpnb{\ensuremath{\mathsf{VP}_{\mathsf{nb}}}}
\newcommand\vpnbzero{\ensuremath{\mathsf{VP}^0_{\mathsf{nb}}}}
\newcommand\vnp{\ensuremath{\mathsf{VNP}}}
\newcommand\vnpzero{\ensuremath{\mathsf{VNP}^0}}
\newcommand\vnpnb{\ensuremath{\mathsf{VNP}_{\mathsf{nb}}}}
\newcommand\vnpnbzero{\ensuremath{\mathsf{VNP}^0_{\mathsf{nb}}}}
\newcommand\vpip{\ensuremath{\mathsf{V\Pi P}}}
\newcommand\vpipzero{\ensuremath{\mathsf{V\Pi P}^0}}
\newcommand\poly{\ensuremath{\mathsf{poly}}}
\newcommand\zz{\ensuremath{\mathbb{Z}}}
\newcommand\nn{\ensuremath{\mathbb{N}}}
\newcommand\fp{\ensuremath{\mathsf{FP}}}
\newcommand\p{\ensuremath{\mathsf{P}}}
\newcommand\pnu{\ensuremath{\mathbb{P}}}
\newcommand\np{\ensuremath{\mathsf{NP}}}
\newcommand\npnu{\ensuremath{\mathbb{NP}}}
\newcommand\per{\ensuremath{\mathrm{PER}}}
\newcommand\sharpp{\ensuremath{\mathsf{\sharp P}}}
\newcommand\pp{\ensuremath{\mathsf{PP}}}
\newcommand\gapp{\ensuremath{\mathsf{GapP}}}
\newcommand\gapppoly{\ensuremath{\mathsf{GapP/poly}}}
\newcommand\chpoly{\ensuremath{\mathsf{CH/poly}}}
\newcommand\ppoly{\ensuremath{\mathsf{P/poly}}}
\newcommand\ch{\ensuremath{\mathsf{CH}}}

\newcommand\bit{\ensuremath{\mathrm{Bit}}}

\newcommand\hit{\ensuremath{\cal H}}

\begin{abstract}
A polynomial identity testing algorithm must determine whether an input polynomial (given for instance by an arithmetic circuit) is identically equal to 0.
In this paper, we show that a deterministic black-box identity testing algorithm for (high-degree) univariate polynomials would imply a lower bound 
on the arithmetic complexity of the permanent. The lower bounds that are known to follow from derandomization of (low-degree) multivariate identity testing are weaker.

To obtain our lower bound it would be sufficient to derandomize identity testing for polynomials of a very specific norm: sums of products of sparse polynomials with sparse coefficients.
This observation leads to new versions of the Shub-Smale $\tau$-conjecture on
integer roots of univariate polynomials. In particular, we show that a lower bound for the permanent would follow if one could give a polynomial upper bound on
the number of real roots of sums of products of sparse polynomials
(Descartes' rule of signs gives such a bound for sparse polynomials and products thereof).

In this third version of our paper we show that the same lower bound
would follow even if one could only prove a slightly superpolynomial
 upper bound on the number of real roots.
This is a consequence of a new result on reduction to depth 4 for arithmetic circuits which we establish in a companion paper.
We also show that an even weaker bound on the number of real roots would
suffice to obtain a lower bound on the size of depth 4 circuits computing
the permanent.
\end{abstract}

\newpage

\section{Introduction}

A polynomial identity testing algorithm must determine whether an input polynomial (given for instance by an arithmetic circuit) is identically equal to 0.
If randomization is allowed, this problem can be solved efficiently thanks to
the well-known Schwarz-Zippel lemma.
Following Kabanets and Impagliazzo~\cite{KI04}, it has become increasingly clear in recent years that efficient {deterministic} algorithms for polynomial identity testing would imply strong lower bounds (the connection between arithmetic circuit lower bounds and derandomization of polynomial identity testing was foreshadowed in a 30 years old paper by Heintz and Schnorr~\cite{HS82}). This approach to lower bounds was advocated in particular by Agrawal~\cite{Agra05}.

In this paper we show that an efficient black-box deterministic identity 
testing algorithm for univariate polynomials of a very specific form
(namely, sums of products of sparse polynomials with sparse coefficients)
would imply that the permanent does not belong to \vpzero.
This is the class of polynomial families computable by constant-free arithmetic circuits of polynomial size and polynomially bounded formal degree.
It plays roughly the same role for constant-free circuits as the class $\vp$
in Valiant's algebraic version of the P versus NP problem
(in Valiant's original setting, arithmetic circuits can use arbitrary constants from the underlying field~\cite{Burg,Valiant79}).

Compared to~\cite{Agra05,KI04}, 
one originality of the present paper 
is to show that lower bound for multivariate polynomials 
such as the permanent would follow from {\em univariate} 
identity testing algorithms.
Most of the recent work on identity testing (surveyed in~\cite{AS09,Sax09}) has been focused on low-degree
multivariate polynomials.\footnote{Two exceptions are \cite{BHLV09,Koi09}.}
Nevertheless, we believe that the univariate approach is worth exploring 
for at least two reasons.

First, it would lead to stronger lower bounds. Indeed, we show that
black-box derandomization of identity testing implies a lower bound for the
permanent, whereas~\cite[Section~6.2]{Agra05} would only yield lower bounds for polynomials with coefficients computable in \pspace\ (this complexity class was independently defined in~\cite{KoiPe09}, where it is called \vpspace; further results on this class 
and other space-bounded classes in Valiant's model can be found in~\cite{KoiPe07b,Poizat08,MahaRao09}).
The lower bound obtained from~\cite{KI04} would be even weaker, but
could be
obtained from a non-black-box identity testing algorithm.

A second, possibly even more important advantage of the univariate approach
is that it leads to new (and hopefully more tractable) versions
of  Shub and Smale's $\tau$-conjecture.
According to the $\tau$-conjecture, the number of integer roots of a univariate polynomial $f \in \zz[X]$  should be bounded by a polynomial function of its arithmetic circuit size (the inputs to the circuit are the constant 1, or the variable $X$).
It was  shown  by B\"urgisser~\cite{Burg09} that the $\tau$-conjecture implies 
a lower bound for the permanent.
Our main ``hardness from derandomization'' result can be viewed as an improvement of B\"urgisser's result.
Indeed, it follows immediately from our result that to obtain a lower bound for the permanent,
one just has to bound the number of integer roots for sums of products of
sparse polynomials with sparse coefficients 
(rather than for arbitrary arithmetic circuits).
Our strongest version of the $\tau$-conjecture raises the intriguing possibility that tools from real analysis might be brought to bear on this  problem (a bound on the number
of real roots of a polynomial is a fortiori a bound on its number of integer roots).
It is known that this approach cannot work for the original $\tau$-conjecture because the number of real roots of a univariate polynomial can grow exponentially as a function of its arithmetic circuit size: Chebyshev polynomials provide such an example~\cite{SmaleProblems}. A similar example was provided earlier
by Borodin and Cook~\cite{BC76} (but they did not provide an analysis of 
the size of constants used by the corresponding arithmetic circuit).
We conjecture that this behavior is not possible for sums of products of sparse polynomials.

\subsection{Main Ideas}

A hitting set $\hit$ for a set $\cal F$ of polynomials
is a (finite) set of points such that there exists for any non-identically 
zero polynomial $f \in \cal F$ at least one point $a \in \hit$ such that 
$f(a) \neq 0$. Hitting sets are sometimes called ``correct test sequences''~\cite{HS82}.
It is well-known that deterministic constructions of hitting sets 
and black-box deterministic identity testing are two equivalent problems: any hitting set for $\hit$
yields an obvious black-box identity testing algorithm 
(declare that $f \equiv 0$ iff $f$ evaluates to 0 on all the points of $\hit$);
conversely, assuming that $\cal F$ contains the identically zero  polynomial, the set of points queried by a black box algorithm on the input $f \equiv 0$
must be a hitting set for $\cal F$. 

The connection between black-box identity testing and lower bounds is especially apparent for univariate polynomials~\cite{HS82}.
Namely, let $\hit$ be a hitting set for $\cal F$. The polynomial 
\begin{equation} \label{hardpoly}
P=\prod_{a \in \hit} (X-a)
\end{equation}
 cannot belong to $\cal F$ since it is nonzero
and vanishes on $\hit$. The same remark applies to all nonzero multiples of $P$. If $\cal F$ is viewed as some
kind of ``complexity class'', we have therefore obtained a lower bound against $\cal F$ by exhibiting a polynomial $P$ which does not belong to $\cal F$.

In the low-degree multivariate setting the polynomial which plays
the same role is not given by such a simple formula as~(\ref{hardpoly}).
Its coefficients can be obtained by solving an exponential size system of linear equations. This can be done in $\pspace$, explaining why the lower bound in~\cite{Agra05} would be for polynomials with coefficients 
computable in $\pspace$.
By contrast one can show that the coefficients in exponential-size products
such as~(\ref{hardpoly}) are in the counting hierarchy, a subclass 
of $\pspace$. This is the reason 
why we can obtain a lower bound for a polynomial in $\vnp$ (namely, the permanent) rather than in $\vpspace$ as in~\cite[Section~6.2]{Agra05}.

It remains to explain why we only have to derandomize 
identity testing for sums of products of sparse polynomials in order to obtain a lower bound.
This class of polynomials comes into the picture thanks to the recent depth reduction theorem of Agrawal and Vinay~\cite{AgraVinay08}: any multilinear polynomial in $n$ variables which has an arithmetic circuit of size $2^{o(n)}$ 
also has a depth-4 arithmetic circuit of size $2^{o(n)}$.
Sums of products of sparse polynomials are very far from being multilinear
(they are univariate polynomials of possibly very high degree).
They are nonetheless connected to depth-4 circuits by a simple 
transformation: if we replace the input variables in a depth-4 circuit
by powers of a single variable $X$, we obtain a SPS polynomial $f(X)$.

At this point, we should stress that we do {\em not} claim that
univariate arithmetic circuits can be efficiently converted into SPS polynomials (this would be a kind of high-degree analogue of Agrawal and Vinay's 
depth reduction theorem). On the contrary, we conjecture that 
such a transformation is in general impossible, and that Chebyshev polynomials
provide a counterexample (because, as pointed out earlier, 
they have too many real roots).
Nevertheless, to obtain our results we represent efficiently (in Theorem~\ref{rep_theorem}) certain exponential size products by sums of products of sparse polynomials. {\em This is possible only under the assumption that the permanent is easy.} This assumption (and the resulting representation) is of course very likely to be false, but there is no harm in making it since the ultimate goal is
a proof by contradiction that the permanent is hard.

\subsection{Organization of the Paper} \label{org}

In the next section we present our model of computation for the permanent (constant-free arithmetic circuits) as well as the corresponding complexity classes.
Then we  recall some definitions and results about the counting hierarchy 
(as explained above, this class plays a crucial role in the derivation
of a lower bound for the permanent).
Finally, we present the result by Agrawal and Vinay on reduction to depth four
for arithmetic circuits, as well as a new result along the same lines~\cite{Koi4d}.

In Section~\ref{sps_section} we define precisely the notion of sum of products of sparse polynomials with sparse coefficients, and explain the connection
to depth-4 circuits.

In Section~\ref{gen_section} we present the notion of algebraic number generator. 
This is basically just a sequence of efficiently computable polynomials in $\zz[X]$.
We wish to use them to construct hitting sets, by taking the sets of all roots
of the polynomials in an initial segment of this sequence.
In Section~\ref{lb_section} we prove our main result: if a polynomial-size 
initial segment provides a hitting set against sums of products of sparse polynomials with sparse coefficients, then the permanent is not in \vpzero.
In fact, using our new result on reduction to depth four~\cite{Koi4d}
we can show that the same lower bound would follow even if the hitting
sets are of slightly superpolynomial size.

In Section~\ref{real_section} we present three new versions of the $\tau$-conjecture, including a ``real $\tau$-conjecture''. A proof of any of these conjectures would yield a lower bound for the permanent. 
We show that a  fairly weak version of the real $\tau$-conjecture would
suffice to obtain a lower bound on the size of depth 4 circuits computing
the permanent.
We conclude the paper with a few remarks on 
some  tools that might be useful to attack these conjectures.

\section{Preliminaries}

\subsection{Complexity of Arithmetic Computations} \label{arith}

We recall that an  arithmetic circuit contains addition, subtraction and multiplication gates. We usually assume that these gates have arity 2, except when dealing with constant-depth circuits as in e.g. Theorem~\ref{reduction}.
The input gates are labelled by variables or constants.
A circuit where the only constants are from the set $\{0, -1, 1\}$ is said to 
be constant-free (in such a circuit one can even assume that $-1$ is the only
constant, and that there are no subtraction gates).
A constant-free circuit represents a polynomial in $\zz[X_1,\ldots,X_n]$, where
$X_1,\ldots,X_n$ are the variables labelling the input gates.

In this paper we investigate the complexity of computing the permanent 
polynomial with constant-free arithmetic circuits.
This model of computation was systematically studied by Malod~\cite{Malod03}.
In particular, he defined a class $\vp^0$ of polynomial families
that are ``easy to compute'' by constant-free arithmetic circuits.
First we need to recall the notion of formal degree:
\begin{itemize}
\item[(i)] The formal degree of an input gate is equal to 1.
\item[(ii)]
The formal degree of an addition or subtraction gate is the maximum of
the formal degrees of its two incoming gates, and the  formal degree
of a multiplication gate is the sum of these two formal degrees.
\end{itemize}
Finally, the formal degree of a circuit is equal to the formal degree
of its output gate.
This is obviously an upper bound on the degree of the polynomial
computed by the circuit.
\begin{definition}
A sequence $(f_n)$ of polynomials 
belongs to $\vp^0$ if there exists a polynomial $p(n)$ 
and a sequence $(C_n)$ of
constant-free arithmetic circuits 
such that $C_n$ computes $f_n$ and is of size (number of gates) 
and formal degree at most $p(n)$.
\end{definition}
The size constraint implies in particular that $f_n$ depends on
polynomially many variables. The constraint on the formal degree 
forbids the computation of polynomials of high degree such as e.g. $X^{2^n}$;
it also forbids the computation of large constants such as $2^{2^n}$.

A central question in the constant-free setting is whether the permanent family belongs to \vpzero. A related question is whether $\tau(\per_n)$, the constant-free arithmetic circuit of the $n \times n$ permanent, is polynomially bounded in $n$.
Obviously, if $\per \in \vpzero$ then  $\tau(\per_n)$ is polynomially bounded in $n$, but (as pointed out in e.g.~\cite{Burg}) it is not clear whether the converse holds true. In this paper we focus on the first question (see section~\ref{remarks} for further comments).

Another important complexity class in the constant-free setting is the
class  $\vnp^0$  of easily definable families. 
It is obtained from $\vp^0$ in the natural way:
\begin{definition}
A sequence $(f_n(X_1,\ldots,X_{u(n)}))$ belongs to $\vnp^0$ if there
exists a sequence $(g_n(X_1,\ldots,X_{v(n)}))$ in $\vp^0$
such that:
$$f_n(X_1,\ldots,X_{u(n)}) = 
\sum_{\overline{\epsilon} \in \{0,1\}^{v(n)-u(n)}}
g_n(X_1,\ldots,X_{u(n)},\overline{\epsilon}).$$
\end{definition}
For instance, the permanent family is in \vnpzero.
If this family in fact belongs to $\vpzero$ then the same is true of every $\vnp^0$ family up to constant multiplicative factors.
Indeed, we have the following result (Theorem~4.3 of~\cite{Koi04}):
\begin{theorem} \label{complete}
Assume that the permanent family is in $\vp^0$. For every family $(f_n)$ in $\vnp^0$ there exists a polynomially bounded function $p(n)$ such that the family $(2^{p(n)}f_n)$ is in $\vp^0$.
\end{theorem}
The occurence of the factor $2^{p(n)}$ in this theorem is due to the fact that the completenes proof of the permanent uses the constant 1/2.
As in~\cite{Koi04} one could avoid this factor by working with the Hamiltonian polynomial instead of the permanent.

The next lemma is Valiant's criterion. The present formulation is basically
that of~\cite[Th.~2.3]{Koi04} but this lemma essentially goes back to~\cite{Valiant79}(see also~\cite[Prop.~2.20]{Burg}).
\begin{lemma}[Valiant's criterion] \label{criterion}
Suppose that $n \mapsto p(n)$ is a
polynomially bounded function, 
and that $f:\nn \times \nn \rightarrow \zz$ is such 
that the map $1^n0j \mapsto f(j,n)$ is in the
complexity class $\gapp/\poly$. 
Then the family $(f_n)$ of multilinear polynomials defined by
\begin{equation} \label{vceq}
f_n (X_1,\ldots,X_{p(n)} )
= \sum_{j \in \{0,1\}^{p(n)}} f(j,n) X_1^{j_1}\cdots
X_{p(n)}^{j_{p(n)}}
\end{equation}
is in $\vnp^0$. Here $j_k$ denotes the bit of $j$ of weight $2^{k-1}$.
\end{lemma}
Note that we use a unary encoding for $n$ but a binary encoding for $j$.
We recall the definition of $\gapp/\poly$ (and a few other boolean complexity classes) in Section~\ref{counting}. In this paper we only need to apply Valiant's criterion to boolean-valued functions 
($f(j,n) \in \{0,1\}$ for all $j$ and $n$)
such that the map $1^n0j \mapsto f(j,n)$ is in \ppoly.

\subsection{The Counting Hierarchy} \label{counting}

A connection between the counting hierarchy and algebraic complexity theory
was discovered in~\cite{ABKB09}. This connection was further explored in~\cite{Burg09} and~\cite{KoiPe10}. For instance, it was shown in~\cite{Burg09} that the polynomials $\prod_{i=0}^{2^n}(X-i)$ have polynomial-size circuits if the the same is true for the permanent family.

We first recall the definition of the two counting classes $\sharpp$ and $\gapp$.\begin{definition}
The class \sharpp\ is the set of functions
    $f:\{0,1\}^*\rightarrow \nn$ such that there exist a language
    $A\in\p$ and a polynomial $p(n)$
    satisfying $$f(x)=\#\{y\in\{0,1\}^{p(|x|)}:(x,y)\in A\}.$$
  A function $f:\{0,1\}^*\rightarrow \zz$ is in \gapp\ if it is the difference of two \sharpp\ functions.
\end{definition}
The counting hierarchy introduced in~\cite{wagner1986} is a class of languages rather than functions. It is defined via the
majority operator $\mathbf{C}$ as follows.
\begin{definition}
If $K$ is a complexity class, the class $\mathbf{C}.K$ is the set of
    languages $A$ such that there exist a language $B\in K$ and a polynomial
    $p(n)$ satisfying
    $$x\in A\iff \#\{y\in\{0,1\}^{p(|x|)}:(x,y)\in B\}\geq 2^{p(|x|)-1}.$$
The $i$-th level $\mathsf{C_iP}$ of the counting hierarchy is defined
    recursively by $\mathsf{C_0P}=\p$ and
    $\mathsf{C_{i+1}P}=\mathbf{C}.\mathsf{C_iP}$. The counting hierarchy \ch\
    is the union of the levels $\mathsf{C_iP}$ for all $i \geq 0$.
\end{definition}
 The counting hierarchy contains all the polynomial
hierarchy $\mathsf{PH}$ and is contained in \pspace.

The arithmetic circuit classes defined in Section~\ref{arith} are nonuniform.
As a result, we will actually work with nonuniform versions of the counting classes defined above. We use the standard Karp-Lipton notation~\cite{KaLi82}:
\begin{definition}
  If $K$ is a complexity class, the class $K/\poly$ is the set of languages
  $A$ such that there exist a language $B\in K$, a polynomial $p(n)$ and a
  family $(a_n)_{n\geq 0}$ of words (the ''advice'')  satisfying
  \begin{itemize}
  \item for all $n\geq 0$, $|a_n|\leq p(n)$;
  \item for all word $x$, $x\in A\iff (x,a(|x|))\in B$.
  \end{itemize}
  Note that the advice only depends on the size of $x$.
\end{definition}

The next lemma \cite[Lemmas~2.6 and~2.13]{Burg09} 
provides a first link between 
arithmetic complexity and the counting hierarchy.
\begin{lemma}\label{lemma_ch}
  If the permanent family is in $\vp^0$ 
  then $\chpoly=\p/\poly$.
\end{lemma}
In particular, Lemma~\ref{lemma_ch} was used to 
show that large sums and products are computable in the counting hierarchy
\cite[Theorem~3.10]{Burg09}. 

In the remainder of this section we summarize some relevant results
from~\cite{KoiPe10}.
\begin{definition} \label{coeff_def}
Let $(f_n)$ be a family of polynomials in $\zz[X]$  such that 
the degree of $f_n$ and the bitsize of its coefficients are smaller than $2^{p(n)}$ for some polynomial $p$. 

The coefficient sequence 
of $(f_n)$ is the (double) sequence of integers $a(n,\alpha)$
defined by the relation
$$f_n(x)=\sum_{\alpha=0}^{2^{p(n)}-1} a(n,\alpha)x^{\alpha}.$$

The coefficient sequence is said to be definable in $\chpoly$ if the language
 ${\rm Bit}(a)=\{(1^n,\alpha,j,b);\ \mbox{ the $j$-th bit of $a(n,\alpha)$ is equal to }b\}$ is in \chpoly.
\end{definition}
Note that in the above definition of  ${\rm Bit}(a)$, the input $n$ is given 
in unary but $\alpha$ and $j$ are in binary (this is the same convention as in~\cite{KoiPe10}; by contrast, in~\cite{Burg09} all inputs are in binary).
\begin{definition}
Let $(f_n)$ be a family of polynomials as in Definition~\ref{coeff_def}.
We say that this family can be evaluated in \chpoly\ if the language
$$\{(1^n,i,j,b);\ 0 \leq i < 2^{p(n)} \mbox{and the $j$-th bit of $f_n(i)$ is equal to }b\}$$ is in \chpoly.
\end{definition}

The following result establishes a connection between these two definitions.
It is stated (and proved) in the proof of the main theorem (Theorem~3.5) of~\cite{KoiPe10}.
\begin{theorem} \label{coeffseq}
Let $(f_n)$ be a family of polynomials as in Definition~\ref{coeff_def}.
If $(f_n)$ can be evaluated in \chpoly\ at integer points,
the coefficient sequence of $(f_n)$ is definable in \chpoly.
\end{theorem}
In~\cite{KoiPe10} we actually prove a multivariate version of this result,
but the univariate case will be sufficient for our purposes.
%

\subsection{Sums of Products of Dense Polynomials} \label{SPD}

Agrawal and Vinay have shown that
polynomials of degree $d=O(m)$ in $m$ variables
 which admit nontrivial arithmetic circuits also admit nontrivial
 arithmetic circuits 
of depth four~\cite{AgraVinay08}. Here, ``nontrivial'' means of size 
$2^{o(d+d\log{m \over d})}$. The resulting depth 4 circuits are $\sum \prod \sum \prod$ arithmetic formulas: the output gate (at depth 4) and the gates at depth 2 are addition gates,  and the other gates are multiplication gates.
This theorem shows that for problems such
as arithmetic circuit lower bounds or black-box derandomization 
of identity testing, the case of depth four circuits is in a certain sense
the general case.

We will need to apply reduction to depth four to 
multilinear polynomials only.
In this case their result (Corollary~2.5 in~\cite{AgraVinay08}) reads as follows:
\begin{theorem}[Reduction to depth four] \label{reduction}
A multilinear polynomial in $m$ variables which has an arithmetic circuit of size $2^{o(m)}$ 
also has a depth 4 arithmetic circuit of size $2^{o(m)}$.
\end{theorem}

But what if we start from arithmetic circuits of size smaller than $2^{o(m)}$ 
(for instance, of size polynomial in $m$) ?
It is reasonable to expect that the size of the corresponding 
depth four circuits will be reduced accordingly, but such a result
cannot be found in~\cite{AgraVinay08}.
We can however prove the following result~\cite{Koi4d}.
\begin{theorem} \label{vpzero2depth4}
Let $(f_n)$ be a $\vpzero$ family of polynomials of degree $d_n=\deg(f_n)$.
This family can be computed by a family $(\Gamma_n)$ of depth four circuits
with $n^{O(\log d_n)}$ addition gates and $n^{O(\sqrt{d_n}\log d_n)}$
multiplication gates.
The family $(f_n)$ can also be computed by a family $(F_n)$ of depth four
arithmetic formulas of size $n^{O(\sqrt{d_n}\log d_n)}$.
The inputs to $\Gamma_n$ and $F_n$ are variables of $f_n$ 
or relative integers 
of polynomial bit size; their multiplication gates are of fan-in 
$O(\sqrt{d_n})$.
\end{theorem}
For instance, if the permanent is in $\vpzero$ it can be computed by depth four
arithmetic formulas of size $n^{O(\sqrt{n}\log n)}$.
Compared to~\cite{AgraVinay08}, 
there are mainly two new elements in Theorem~\ref{vpzero2depth4}:
\begin{itemize}
\item[(i)] The size bounds for the depth-four circuits $(\Gamma_n)$ 
and $(F_n)$.
\item[(ii)] The bit size bound for the inputs of these circuits.
\end{itemize}
Our main results rely on these two new elements. In particular, we use (i)
to show that constructing hitting sets of slightly superpolynomial size
will still imply that the permanent is not in $\vpzero$.
To the author's knowledge, an analysis of the size of constants 
created in the depth-reduction procedure
of~\cite{AgraVinay08} has not been carried out yet.

We can formulate Theorem~\ref{vpzero2depth4}
 in more traditional mathematical language.
\begin{corollary} \label{vpzero2SPS}
Let $(f_n)$ be a $\vpzero$ family of polynomials of degree $d_n=\deg(f_n)$.
Each $f_n$ can be represented by an expression of the form 
$\sum_{i=1}^k \prod_{j=1}^m f_{ij}$ where $k=n^{O(\sqrt{d_n}\log d_n)}$
and 
$m=O(\sqrt{d_n})$.
The $f_{ij}$ are polynomials of degree 
$O(\sqrt{d_n})$ and their coefficients are relative integers of polynomial
bit size. Moreover, the sum of the number of monomials in all the $f_{ij}$ 
is $n^{O(\sqrt{d_n}\log d_n)}$,
and there are only $n^{O(\log d_n)}$ distinct $f_{ij}$.
\end{corollary}
\begin{proof}[Sketch]
Each multiplication gate at depth 1 in the depth four circuit of 
Theorem~\ref{vpzero2depth4} computes a monomial. Each addition gate at 
depth 2 computes a $f_{ij}$. A multiplication gate at depth 3 computes
an expression of the form $\prod_j f_{ij}$. The output gate computes the
final sum. 

In fact, several multiplication gates at depth 1 may contribute to the 
same monomial of a $f_{ij}$ and the monomial will be obtained as the 
sum of the outputs of these multiplication gates.\footnote{This is bound to happen since as a polynomial of degree $O(\sqrt{d_n})$ in $n^{O(1)}$ variables
$f_{ij}$ can have at most $n^{O(\sqrt{d_n})}$ monomials, but there are many more
multiplication gates.}
Taking this sum preserves
the polynomial size bound on coefficients since there are only 
$n^{O(\sqrt{d_n}\log d_n)}$ multiplication gates.
\end{proof}
A polynomial is sparse if it has few monomials compared to the maximal
number of monomials possible given its degree and number of variables
(recall that for a polynomial in $n$ variables 
of degree $d$, this number is 
$n+d \choose d$). There is no reason for the $f_{ij}$ to be sparse in general
(but they have few terms compared to the maximum possible for $f_n$).
As explained in the next section, if we replace the variables of the $f_{ij}$ by a quickly growing sequence
of powers of a single variable $X$, we obtain truly sparse 
univariate polynomials.

\section{Sums of Products of Sparse Polynomials} \label{sps_section}

A sums of products of sparse polynomials is an expression of the form $\sum_i \prod_j f_{ij}$ where each $f_{ij} \in \zz[X]$ is a sparse univariate polynomial. Here ``sparse'' means as usual that we only represent the nonzero monomials of each $f_{ij}$.
As a result one can represent concisely polynomials of very high degree. We define the size of such an expression as the
sum of the number of monomials in all the $f_{ij}$. Note that this measure of size does not take into account the size of 
the coefficients of the $f_{ij}$, or their degrees. These relevant parameters are taken into account in the following definition.
\begin{definition} \label{SPSdef}
We denote by ${\rm SPS}_{s,e}$ the set of all polynomials in $\zz[X]$ which can be represented by an expression of the form $\sum_i \prod_j f_{ij}$ so that:
\begin{itemize}
\item The size of the expression as defined above is at most $s$.
\item Each coefficient of each $f_{ij}$ can be written as the difference of two nonnegative integers with at most $s$ nonzero digits 
in their binary representations.
\item These coefficients are of absolute value at most $2^e$, and the $f_{ij}$ are of degree at most $e$.
\end{itemize}
\end{definition}

\begin{remark}
The polynomials $f_{ij}$ in this definition
can be thought of as "sparse polynomial with sparse coefficients". The integer $s$ serves as a sparsity parameter for
the number of monomials as well as for the number of digits in their 
coefficients. 
A typical choice for these parameters is $s=2^{o(n)}$ and $e=2^{O(n)}$, where 
$n$ represents an input size (see for instance Theorem~\ref{rep_theorem} in Section~\ref{lb_section}).
\end{remark}
We will show in Section~\ref{lb_section} that constructing polynomial size hitting sets for sums of products of sparse polynomials implies the lower bound 
$\per {\not \in} \vp^0$. 
Here ``polynomial size'' means polynomial in $s +\log e$.
It is quite natural to insist on a size bound which is polynomial in $s$ and $\log e$: $s$ is an arithmetic circuit size bound, 
and $\log e$ can also be interpreted as an arithmetic cost since each power $x^{\alpha}$ in an $f_{ij}$ can be computed 
from $x$ in $O(\log e)$ operations by repeated squaring. 
Likewise, we can write each coefficient of each $f_{ij}$ as the difference 
of two nonnegative integers as in Definition~\ref{SPSdef}, 
and each of the $\leq s$ powers of 2 occuring in a nonnegative integer
can be computed from the constant 2 in $O(\log e)$ operations.
Each coefficient can therefore be computed in $O(s\log e)$ operations.
As a result, a polynomial in ${\rm SPS}_{s,e}$ can be evaluated from the constant 1 and the variable $X$ in a number of arithmetic operations which is polynomial in $s+\log e$.

The size of a SPS polynomial as we have defined it is essentially the size of a depth three  arithmetic circuit (or more precisely of a depth three arithmetic formula) computing the polynomial. In this depth three formula each input gate carries a monomial; each addition gate at level 1 computes a $f_{ij}$; each multiplication gate at level 2 computes a product of the form $\prod_j f_{ij}$;
and the output gate at level 3 computes the final sum.

We can further refine this representation of SPS polynomials 
by arithmetic formulas. Namely, instead of viewing the monomial $aX^{\beta}$ as
an atomic object which is fed to an input gate, we can decompose it as a sum of terms of the form $\pm 2^{\alpha}X^{\beta}$; and each term can be further decomposed as a product of factors of the form $\pm 2^{2^i}$ and $X^{2^j}$.
The resulting object is a depth four formula where each input gate carries an expression of the form $\pm 2^{2^i}$ or $x^{2^j}$ (note the symmetry between variables and constants in this representation).
This connection between depth four formulas and SPS polynomials plays a crucial role in our results.
In particular, we will use the following result in Section~\ref{lb_section}.
\begin{proposition} \label{subs_prop}
Let $(f_n(\overline{x},\overline{z}))$ be a $\vpzero$ family 
of multilinear polynomials,
with $\overline{x}$ and $\overline{z}$ two tuples of variables of length
$c \cdot n$ each (for some constant~$c$). 
Let $f'_n(x)$ be the univariate polynomial defined from $f_n$ by the substitution:
\begin{equation} \label{univ_subs}
f'_n(x)=f_n(x^{2^0},x^{2^1},\dots,x^{2^{c\cdot n-1}},2^{2^0},2^{2^1},\dots,2^{2^{c\cdot n-1}}).
\end{equation}
The $f'_n$ belong to ${\rm SPS}_{s,e}$ where $s=n^{O(\sqrt{n}\log n)}$
and $e=2^{O(n)}$.

More precisely, each $f'_n$ can be represented by an expression of the form 
$\sum_{i=1}^k \prod_{j=1}^m f'_{ij}$ where $k=n^{O(\sqrt{n}\log n)}$
and $m=O(\sqrt n)$. The $f'_{ij}$ are polynomials of degree
$2^{O(n)}$ and have at most $n^{O(\sqrt{n})}$ nonzero monomials.  
Each coefficient of a monomial can be written as the difference  
of two non-negative integers of bit size $2^{O(n)}$ with at most
$n^{O(\sqrt{n})}$ nonzero digits. 
Moreover,  the sum of the number of monomials in all the $f'_{ij}$ 
is $n^{O(\sqrt{n}\log n)}$, 
and there are only $n^{O(\log n)}$ distinct~$f'_{ij}$.
\end{proposition}
\begin{proof}[Sketch]
This is a fairly straightforward consequence of Corollary~\ref{vpzero2SPS}.
In particular, we have at most $n^{O(\sqrt{n})}$ monomials in $f'_{ij}$ 
because this is also
an upper bound on the number of monomials in the corresponding polynomials
$f_{ij}$ of Corollary~\ref{vpzero2SPS}. 
The effect of multiplication by the powers
of two in~(\ref{univ_subs}) is to shift the coefficients of the $f_{ij}$ 
without increasing their bit size, and we need to add (and subtract)  
$n^{O(\sqrt{n})}$ shifted coefficients to obtain a coefficient of a $f'_{ij}$.
\end{proof}

\section{Algebraic number generators} \label{gen_section}

As explained in Section~\ref{org}, we wish to construct hitting sets by taking
the sets of all roots of the polynomials in an initial segment
of an efficiently computable sequence of polynomials.
The following definition makes the notion of ``efficiently computable'' 
precise (compare with the notion of {\em hitting set generator} in ~\cite[Section~6.2]{Agra05}).
\begin{definition} \label{gendef}
An algebraic number generator is a sequence $(f_i)_{i \geq 1}$  of nonzero univariate polynomials $f_{i}(X)=\sum_{\alpha} a(\alpha,i)X^{\alpha}$
such that for some integer constant $c \geq 1$:
\begin{enumerate}
\item The exponents $\alpha$ range from 0 to $i^c$;
 \item $a(\alpha,i)$ is a sequence of integers of absolute
  value $ \leq 2^{i^c}$;
  \item The language $L(f)=\{(\alpha,i,j,b);\ \mbox{ the $j$-th bit of $a(\alpha,i)$ is equal to }b\}$ is in \chpoly. 
\end{enumerate}
\end{definition}

In the above definition we work with the complexity class $\chpoly$ because this is the largest complexity class for which our proofs go through. As shown in the next example, the language $L(f)$ can often be located in a much smaller complexity class.
\begin{example}
Each of the three sequences $f_{i}=x-i$, $(x^i-1)$ or $x^i-2^ix+i^2+1$ 
is an algebraic number generator. 
Notice that in these three examples we can compute the coefficients of the $f_{i}$ in polynomial time rather than in $\chpoly$, i.e., there is no need for counting and the construction of the $f_i$ is uniform.
\end{example}
\begin{theorem} \label{bigproduct}
Let $(f_i)$ be an algebraic number generator.
From this sequence we define a  family of univariate polynomials $g_n$ by the formula:
$$g_n(x)=\prod_{i=1}^{2^n}f_i(x).$$
The coefficient sequence $b(n,\alpha)$ of $g_n$,
defined by $g_n(x)=\sum_{\alpha}b(n,\alpha)x^{\alpha}$,
is definable in $\ch/\poly$.
\end{theorem}
\begin{proof}
The family $(g_n)$ can be evaluated in \chpoly\ at integer points.
This follows from the fact that integer sequences definable in \chpoly\ 
are stable under products and summations~\cite[Theorem~3.10]{Burg09}.
The result then follows from Theorem~\ref{coeffseq}.
\end{proof}
We illustrate this result on two examples.
\begin{example}
For $f_i = x-i$ we have $g_n(x)=\prod_{i=1}^{2^n}(x-i)$. This is the Pochhammer-Wilkinson polynomial of order $2^n$.
It was shown in~\cite[proof of Main Theorem~1.2]{Burg09} that the coefficient sequence of Pochhammer-Wilkinson polynomials is definable in $\ch$.
\end{example}
\begin{example}
For $f_i = x^i-1$ we have $g_n(x)=\prod_{i=1}^{2^n}(x^i-1)$. 
This product can be written as 
\begin{equation} \label{prodeq} 
g_n(x)=\prod_{\overline{\epsilon}}h_n(x,\overline{\epsilon})
\end{equation}
 where the auxiliary family $h_n$ is defined by: $$h_n(x,\epsilon_1,\ldots,\epsilon_n)=x\prod_{j=1}^n[(1-\epsilon_j)+\epsilon_j x^{2^{j-1}}]-1.$$
Note that the powers $x^{2^{j-1}}$ in the above formula can be computed efficiently by repeated squaring.
The family $(h_n)$ therefore belongs to the class $\vp^0_{nb}$ of polynomials that can be evaluated in a polynomial number of arithmetic operations in the constant-free unbounded-degree model. It then follows from~(\ref{prodeq}) that $g_n$ belongs to the class $\vpip^0$
(by definition, the families of this class are obtained as in~(\ref{prodeq}) from a $\vp^0_{nb}$ family by taking an exponential-size product over a $\vp^0_{nb}$ family).
It is shown in~\cite[Theorem~3.7]{KoiPe10} that the class $\vpip^0$ would collapse to $\vp^0_{nb}$ if $\vnp^0$
collapses to $\vp^0$. The proof of this theorem is based on definability of coefficients in $\chpoly$ for $\vpip^0$ families (in our particular example there is again no need for nonuniformity since the family $(f_i)$ is uniform).
\end{example}

\section{From a Hitting Set to a Lower Bound}
\label{lb_section}

In this section we prove our main result: constructing hitting sets for the class ${\rm SPS}_{s,e}$ of sums of products of sparse polynomials with sparse coefficients implies a lower bound for the permanent (recall that the class  ${\rm SPS}_{s,e}$ 
is defined in Section~\ref{sps_section}).

We begin with a lemma showing that under the assumption $\per \in \vpzero$,
polynomials with coefficients definable in \chpoly\ can be efficiently represented by sums of products of sparse polynomials.
This result  is an adaptation of~\cite[Lemma~3.2]{KoiPe10}, which was itself a scaled up version of~\cite[Th.~4.1(2)]{Burg09}. The main new ingredient 
is reduction to depth four as presented in Section~\ref{SPD}.
\begin{lemma} \label{chtosparse}
  Let $g_n(x)=\sum_{\alpha} a(n,\alpha) x^{\alpha}$
  where the integers $\alpha$ range from 0 to $2^{c\cdot n}-1$,
  $a(n,\alpha)$ is a sequence of integers of absolute
  value $< 2^{2^{c\cdot n}}$ definable in \chpoly, and $c$ is an integer constant (independent of $n$).

  If $\per \in \vp^0$ 
  there is a polynomially bounded function $p(n)$ such that 
$2^{p(n)}g_n \in {\rm SPS}_{s,e}$ where $s=n^{O(\sqrt{n}\log n)}$ 
and $e=2^{O(n)}$. 
\end{lemma}

\begin{proof}  Expand $a$ in binary: $\displaystyle a(n,\alpha)=\sum_{i=0}^{2^{c\cdot n}-1}
  a_i(n,\alpha)2^i.$ Let $h_n$ be the following multilinear polynomial:
  $$h_n(x_{1},x_{2},\dots,x_{c\cdot n},z_1,\dots,z_{c\cdot n})=\sum_{i=0}^{2^{c\cdot n}-1}\sum_{\alpha=0}^{2^{c\cdot n}-1}a_i(n,\alpha) z_1^{i_1}\cdots
  z_{c\cdot n}^{i_{c\cdot n}} x_{1}^{\alpha_1}x_{2}^{\alpha_2}\cdots
  x_{c\cdot n}^{\alpha_{c\cdot n}}.$$
In this formula, the exponents $i_j$ and $\alpha_j$ denote the binary digits
of the integers $i$ and $\alpha$.
  Then we have:
\begin{equation} \label{plug}
h_n(x^{2^0},x^{2^1},\dots,x^{2^{c\cdot n-1}},2^{2^0},2^{2^1},\dots,2^{2^{c\cdot n-1}})=g_n(x).
\end{equation}
Assume that the permanent family is in $\vp^0$. by Lemma~\ref{lemma_ch} the nonuniform counting
  hierarchy collapses, therefore computing the $i$-th bit $a_i(n,\alpha)$
  of $a(n,\alpha)$ on input $(1^n,\alpha,i)$ is in \gapppoly\ (and
  even in \p/\poly). By Lemma~\ref{criterion}, $(h_n)\in\vnpzero$. 
  By Theorem~\ref{complete} there exists a polynomially bounded function $p(n)$ 
such that the family $f_n=2^{p(n)}h_n$ is in $\vp^0$. 
Applying Proposition~\ref{subs_prop} to $(f_n)$ shows that the polynomials $f'_n=2^{p(n)}g_n$ are in ${\rm SPS}_{s,e}$ for $s=n^{O(\sqrt{n}\log n)}$ and $e=2^{O(n)}$.
\end{proof}
Next we we show that the product of the first $2^n$ polynomials 
of an algebraic number generator can be represented by
a sum of products of sparse polynomials of subexponential size, 
assuming again that the permanent is in \vpzero. 
\begin{theorem} \label{rep_theorem}
  Let $(f_i)$ be an algebraic number generator and $g_n(x)=\prod_{i=1}^{2^n} f_{i}(x)$.
  If $\per \in \vpzero$ there is a polynomially bounded function $p(n)$ such that $2^{p(n)}g_n \in {\rm SPS}_{s,e}$ where $s=n^{O(\sqrt{n}\log n)}$ and $e=2^{O(n)}$.
Here ${\rm SPS}_{s,e}$ is the class of sums of products of sparse polynomials
from Definition~\ref{SPSdef}.
\end{theorem}


\begin{proof}
We wish to apply Lemma~\ref{chtosparse} 
to the polynomial $g_n(x)=\prod_{i=1}^{2^n} f_{i}(X)$.
Each polynomial $f_{i}$ in this product is of degree less than $2^{cn}$ (except possibly $f_{n}$, which may be of degree up to $2^{cn}$). Hence  $g_n$ is of degree less than $2^{(c+1)n}$.
As to the coefficient size, we have $||g_n||_1 \leq \prod_{i}||f_{i}||_1$ where $||.||_1$ denotes the sum
of the absolute values of the coefficients of a polynomial. For each $i$ we have $||f_{i}||_1 < (2^{cn}+1)\cdot 2^{2^{cn}} \leq 2^{2^{(c+1)n}}$ so that $||g_n||_1 \leq 2^{2^{(c+2)n}}$.
Finally, definability of coefficients in $\chpoly$ is provided by Theorem~\ref{bigproduct}.\end{proof}


We can finally prove our main result.
\begin{theorem}[Lower Bound from Hitting Sets] \label{lbth}
Let $(f_{i})$ be an algebraic number generator and $H_{m}$ the set of  all roots of the polynomials $f_{i}$ for all $i \leq m$. Let  $q$ and $r$ be two functions such that $H_{q(s)+r(e)}$
is a hitting set for ${\rm SPS}_{s,e}$. The permanent is not in $\vp^0$ 
if $r(e)=e^{o(1)}$ and $q$ satisfies the following condition: for some constant
$c <1$ and $s$ large enough, $q(s) \leq 2^{(\log s)^{1+c}}$.
\end{theorem}
The conditions on $q$ and $r$ cover in particular the case of hitting sets
of size polynomial in $s$ and $\log e$. This special case was treated
in an earlier version of this paper\footnote{\url{http://arxiv.org/abs/1004.4960v2}}. Note also that any set of more than $s\cdot e$ complex numbers is a hitting
set since any polynomial in ${\rm SPS}_{s,e}$ is of degree at most $s\cdot e$.
\begin{proof}[of Theorem~\ref{lbth}]
Let $g_n(x)=\prod_{i=1}^{2^n} f_{i}(x)$ be the polynomial of Theorem~\ref{rep_theorem}.
Assume by contradiction that:
\begin{itemize}
\item[(i)] There exists functions  $q$ and $r$ such that $H_{q(s)+r(e)}$
is a hitting set for~${\rm SPS}_{s,e}$, where $q$ and $r$ satisfy
the conditions in the statement of the theorem.
\item[(ii)] The permanent family is in $\vp^0$. 
\end{itemize}
From our second assumption and Theorem~\ref{rep_theorem} we know that $2^{p(n)}g_n$ is in ${\rm SPS}_{s,e}$
for $s=n^{O(\sqrt{n}\log n)}$, $e=2^{O(n)}$ and some polynomially bounded function~$p(n)$. The conditions on $q$ and $r$ imply that for these values of 
$s$ and $e$ we have $q(s)+r(e)=2^{o(n)}$.
Hence by (i), for $n$ large enough $H_{2^n}$ is a hitting set for $g_n$.
This is a contradiction since $g_n$ vanishes on the hitting set $H_{2^n}$ 
but is not identically 0.
\end{proof}

\section{Hitting Sets from Real Analysis ?} \label{real_section}

In this section we present our new versions of the $\tau$-conjecture.
Each of the three conjectures implies that the permanent is not in $\vpzero$. 
\begin{conjecture}[$\tau$-conjecture for SPS polynomials] \label{tauSPS}
There is a polynomial $p$ such that any nonzero polynomial in ${\rm SPS}_{s,e}$ has at most $p(s+\log e)$ integer roots.
\end{conjecture}
 This conjecture implies that $\per {\not \in} \vp^0$ (apply Theorem~\ref{lbth} to the algebraic number generator $f_i(x)=x-i$).
Conjecture~\ref{tauSPS} follows from the $\tau$-conjecture of Shub and Smale on integer roots of polynomials~\cite{ShubSmale,SmaleProblems} since, as explained after Definition~\ref{SPSdef}, polynomials in ${\rm SPS}_{s,e}$ can be evaluated
by constant-free arithmetic circuits of size polynomial in $s$ and $\log e$.
It was already shown  in~\cite{Burg09} that the $\tau$-conjecture implies 
a lower bound for the permanent. The point of Conjecture~\ref{tauSPS} is that to obtain such a lower bound we no longer have
to bound the number of integer roots of arbitrary arithmetic circuits: we need only do this for sums of products of sparse polynomials. This looks like a much more manageable class of circuits, but the question is of course still wide open.
Another related  benefit of SPS polynomials in this context is that techniques from real analysis might become applicable.
Before explaining this in more detail we formulate a somewhat stronger conjecture. The idea is that the parameter $e$ in Conjecture~\ref{tauSPS} as well as the sparsity hypothesis on the integer coefficients might be irrelevant.
This leads to:
\begin{conjecture}[$\tau$-conjecture for SPS polynomials, strong form] \label{conj2}
Consider a nonzero polynomial of the form 
$$f(X)=\sum_{i=1}^k \prod_{j=1}^m f_{ij}(X),$$
 where each $f_{ij} \in \zz[X]$ has at most $t$ monomials. The number of integer roots of $f$ is bounded by a polynomial function of $kmt$.
\end{conjecture}
Note that the size of $f$ as defined in Section~\ref{sps_section} is bounded by $kmt$. Therefore, Conjecture~\ref{conj2} is indeed stronger than Conjecture~\ref{tauSPS}.
Finally, we formulate an even stronger conjecture.
\begin{conjecture}[real $\tau$-conjecture] \label{realtau}
Consider a nonzero polynomial of the form  
\begin{equation} \label{spspoly}
f(X)=\sum_{i=1}^k \prod_{j=1}^m f_{ij}(X),
\end{equation}
where each $f_{ij} \in \rr[X]$ has at most $t$ monomials. The number of real roots of $f$ is bounded by a polynomial function of $kmt$.
\end{conjecture}
One could also formulate a weak version of the real $\tau$-conjecture where the parameters $s$ and $e$ would play the same role as in Conjecture~\ref{tauSPS}.
Also, instead of a bound on the number of real roots which is polynomial
in $kmt$ one could seek a bound $q(kmt)$ which is slightly superpolynomial
in the sense of Theorem~\ref{lbth}: 
for some constant $c<1$ and $s$ large enough, we have
$q(s) \leq 2^{(\log s)^{1+c}}$. 
By Theorem~\ref{lbth}, such a bound would still be strong enough 
to conclude that the permanent is not in $\vpzero$ (consider again
the algebraic number generator $f_i(x)=x-i$, and the function $r(e)=1$).
This goal might still be difficult to achieve, so it would be of great interest
to establish upper bounds on the number of real roots that are even weaker
but still strong enough to imply interesting lower bounds.
For instance:
\begin{proposition}
Assume that for nonzero polynomials of the form~(\ref{spspoly}) the number of real roots is less than $q(kmt)$, where the function $q$ 
satisfies the condition $q(s)=2^{s^{o(1)}}$. 
Then the permanent is not computable by polynomial size depth 4 circuits
 using polynomial size integer constants.
\end{proposition}
\begin{proof}
Assume that the permanent is computable by polynomial size depth 4 circuits using polynomial size integer constants. 
In particular, the permanent is in $\vpzero$.
By completeness of the permanent we have the following strengthening of
Theorem~\ref{complete}: for every family  $(h_n)$ in $\vnp^0$ there exists a polynomially bounded function $p(n)$ such that the family $(2^{p(n)}h_n)$ is 
computable by polynomial size depth 4 circuits (using polynomial size
integer constants).

We consider again the algebraic number generator $f_i(x)=x-i$ and the poynomial
$g_n(x)=\prod_{i=1}^{2^n} f_i(x)$ of Theorem~\ref{rep_theorem}.
We claim that $g_n$ can be expressed as a SPS polynomial of size
polynomial in $n$. This yields a contradiction since the assumption
in the statement of the Proposition implies that $g_n$ has
$2^{n^{o(1)}}$ real roots but in reality $g_n$ has $2^n$ 
integer roots.

The proof of the claim is similar to the proof of Lemma~\ref{chtosparse} and Theorem~\ref{rep_theorem}. 
In particular, we have for $g_n$ the same representation as in 
equation~(\ref{plug}) of Lemma~\ref{chtosparse}. 
But now, the above-mentioned strengthening of Theorem~\ref{complete} shows
that $h_n$ is computable by a depth 4 circuit of polynomial size.
We obtain the SPS polynomial for $g_n$ by plugging powers of $x$ 
into this circuit.
\end{proof}

At present there isn't a lot of evidence for or against Conjecture~\ref{realtau}. 
We do know that the conjecture holds true when $k=1$: by Descartes' rule each polynomial $f_{1j}$ has at most $2t-2$ nonzero real roots, so $f$ has at most $2m(t-1)+1$ real roots. 
Also some indirect evidence is provided by the few known examples of polynomials with short arithmetic circuits but many real roots~\cite{BC76,SmaleProblems}:
these examples are definitely not given as sums of products of sparse polynomials.
The case $k=2$ already looks nontrivial. 
In the general case we can expand $f$ as a sum of at most $kt^m$ monomials,
so we have at most $2kt^m-1$ real roots.
A refutation of the conjecture would be interesting from the point of view of real algebra and geometry as it would yield examples of ``sparse like'' polynomials with many real roots. Of course, a proof of the conjecture would be even more interesting as it would yield a lower bound for the permanent.


\section{Final Remarks} \label{remarks}

We have shown that constructing hitting sets for sums of products of sparse
polynomials with sparse coefficients will show that $\per {\not \in} \vp^0$.
It should be possible to obtain a variation of this result where the
conclusion is that $\tau(\per_n)$, the constant-free arithmetic 
circuit complexity of the permanent, is not polynomial in $n$.
To obtain this stronger conclusion, a stronger hypothesis should be necessary.
It seems natural to expect that the role payed by sparse polynomials 
with sparse coefficients will now played by sparse polynomials 
with coefficients of ``small''  $\tau$-complexity (this is a larger class of polynomials since sparse coefficients 
are certainly of small $\tau$-complexity).

Most importantly, one should try to prove or disprove 
the real $\tau$-conjecture. A solution in the case $k=2$ (a sum of two products
of sparse polynomials) would already be quite interesting.
We note that the search for good upper bounds on the number of solutions
of sparse multivariate systems is a topic of current interest in real algebraic
geometry.
The theory of fewnomials~\cite{Khov91} provides finiteness results and sometimes quantitative estimates on the number of real roots in very general ``sparse like'' situations. The general estimates from~\cite{Khov91}, at least when applied in a straightforward manner, do not seem 
strong enough to imply the real $\tau$-conjecture.
Nevertheless, one can hope that the methods developed in~\cite{Khov91} as well as in more recent work such as~\cite{BerBiSot,BiSot,LiRojasWang} will turn out to be useful.

\bibliographystyle{plain}

\end{document}